\documentclass[reqno,12pt]{amsart}
\usepackage{hyperref}
\usepackage{amssymb}
\usepackage{amsmath}
\usepackage{amsthm}
\usepackage{mathrsfs}
\usepackage[arrow,matrix,curve]{xy}
\usepackage{graphicx}
\usepackage{hyperref}
\usepackage{amscd,verbatim}
\usepackage{mathtools}
\usepackage{comment}

\usepackage{fullpage,enumerate}
\newcommand{\bea}{\begin{eqnarray}}
\newcommand{\eea}{\end{eqnarray}}
\newcommand{\beq}{\begin{equation}}
\newcommand{\eeq}{\end{equation}}

\newcommand\wt{\widetilde}

\newtheorem{theorem}{Theorem}[section]
\newtheorem{conjecture}[theorem]{Conjecture}

\newtheorem{proposition}[theorem]{Proposition}

\theoremstyle{remark}

\newtheorem{remark}[theorem]{Remark}

\numberwithin{equation}{section}
\numberwithin{theorem}{section}

\newcommand\cK{\mathcal{K}}

\newcommand\cL{\mathcal{L}}

\newcommand\cP{\mathcal{P}}

\newcommand{\R}{\ensuremath{\mathbb R}}

\def\CT{{CT}}

\def\wt#1{\widetilde{#1}}

\newcommand{\CC}{{\mathbb C}}

\newcommand{\RR}{{\mathbb R}}
\newcommand{\TT}{{\mathbb T}}
\newcommand{\ZZ}{{\mathbb Z}}

\newcommand{\abZ}{\underset{ab}{{\mathbb Z}^2}}
\newcommand{\bcZ}{\underset{bc}{{\mathbb Z}^2}}
\newcommand{\acZ}{\underset{ac}{{\mathbb Z}^2}}

\newcommand{\aZ}{\underset{a}{\mathbb Z}}
\newcommand{\bZ}{\underset{b}{\mathbb Z}}
\newcommand{\cZ}{\underset{c}{\mathbb Z}}
\newcommand{\dZ}{\underset{d}{\mathbb Z}}
\newcommand{\abT}{\underset{ab}{{\mathbb T}^2}}
\newcommand{\bcT}{\underset{bc}{{\mathbb T}^2}}
\newcommand{\acT}{\underset{ac}{{\mathbb T}^2}}

\newcommand{\aT}{\underset{a}{\mathbb T}}
\newcommand{\bT}{\underset{b}{\mathbb T}}
\newcommand{\cT}{\underset{c}{\mathbb T}}
\newcommand{\dT}{\underset{d}{\mathbb T}}

\newcommand{\heisZ}{\mathrm{Heis}^\ZZ}

\newcommand{\bcR}{\underset{bc}{{\mathbb R}^2}}
\newcommand{\acR}{\underset{ac}{{\mathbb R}^2}}
\newcommand{\aR}{\underset{a}{\mathbb R}}
\newcommand{\bR}{\underset{b}{\mathbb R}}
\newcommand{\cR}{\underset{c}{\mathbb R}}

\newcommand{\heisR}{\mathrm{Heis}^\RR}
\newcommand{\nil}{\mathrm{Nil}}

\begin{document}

\title[T-duality trivializes bulk-boundary correspondence]
{T-duality trivializes bulk-boundary correspondence: the parametrised case}

\author[KC Hannabuss]{Keith C. Hannabuss}

\address[Keith Hannabuss]{
Mathematical Institute, 
24-29 St. Giles', 
Oxford, OX1 3LB, and 
Balliol College, 
Oxford, OX1 3BJ,
U.K.} 

\email{kch@balliol.ox.ac.uk}

\author[V Mathai]{Varghese Mathai}

\address[Varghese Mathai]{
Department of Pure Mathematics,
School of  Mathematical Sciences, 
University of Adelaide, 
Adelaide, SA 5005, 
Australia}

\email{mathai.varghese@adelaide.edu.au}

\author[G.C.Thiang]{Guo Chuan Thiang}

\address[Guo Chuan Thiang]{
Department of Pure Mathematics,
School of  Mathematical Sciences, 
University of Adelaide, 
Adelaide, SA 5005, 
Australia}

\email{guo.thiang@adelaide.edu.au}

\begin{abstract}
We state a general conjecture that T-duality trivialises a model for the bulk-boundary correspondence in the parametrised context. We give evidence that it is valid by proving it in a special interesting case, which is relevant both to String Theory and to the study of topological insulators with defects in Condensed Matter Physics.
\end{abstract}

\maketitle


\section{Introduction}

\label{sec:sketch} 
Recently the last two authors introduced T-duality in the study of topological insulators \cite{MT}. As an application, it was shown in \cite{MT1} that T-duality trivializes the \emph{bulk-boundary correspondence} (in a $K$-theoretic sense as pioneered by \cite{Kellendonk1,Kellendonk3}, see also \cite{Prodan,Elbau,Graf,Hatsugai}, and \cite{Bourne} for a $KK$-theory perspective) in two dimensions, even in the presence of disorder. A similar simplification was found in basic cases in higher dimensions, for both complex and real $K$-theory, in a follow-up paper \cite{MT2}.
The bulk-boundary correspondence of topological invariants is an important aspect of the analysis of the integer quantum Hall effect via noncommutative geometry \cite{Bellissard,Connes85,Connes94}, as well as its generalisations \cite{CHMM,Marcolli}. For a recent interesting alternate approach in the general field of topological phases of matter, see \cite{Witten}. 

The notion of a Brillouin torus of quasi-momenta familiar from Bloch theory is a central one in condensed matter physics. Besides being a primary source of topological invariants of physical interest, the torus structure permits one to perform T-duality transformations. This means that the topological invariants, especially the $K$-theoretic ones \cite{FM,T,Loring,Sheinbaum,Kellendonk5}, can be analysed on the T-dual side, or in real space, where they may be more easily understood. In fact, the concept of a Brillouin torus admits a vast generalisation using the language of noncommutative geometry. T-duality, as a geometric version of a generalized Fourier transform, continues to make sense in the more general noncommutative and even parametrised setting.

The first two authors introduced \emph{parametrised} noncommutative strict deformation quantization of $C^*$-algebras with a torus action, and related it to T-duality of principal torus bundles with H-flux \cite{HM09,HM10}. While these papers were initially motivated by dualities in string theory, they may also have applications in condensed matter physics. For instance, parametrised families of physical systems are interesting in that they can have observable geometric phases or holonomies associated to them. In principle, such phenomena can be combined in non-trivial ways with commonly used topological invariants such as those arising from Bloch band topology. 

In this paper, we prove that the bulk-boundary correspondence is again trivialised in the \emph{parametrised} context, for a special case with nontrivial H-flux. In this case, the T-dual is a parametrised deformation of a 3-torus, and may be identified with the integer Heisenberg group algebra. The latter has an interpretation as a deformed version of the Brillouin torus in 3D, and we apply this picture to a model of a topological insulator with a uniform distribution of lattice dislocations. In this condensed matter physics context, a second related trivialization result for the bulk-boundary correspondence is proven. Finally, we also state a general conjecture which subsumes our previous results and is of independent mathematical interest.

\tableofcontents

\section{T-duality and bulk-boundary correspondence}
We begin by briefly summarizing the results in \cite{HM09,HM10}, which will provide the technical background for the general conjecture \ref{generalconjecture}. This section can be skimmed over to get to the main results of this paper in Section \ref{section:Tdualitytrivial} more quickly. For an application to topological insulators with lattice dislocations, see Section \ref{section:screw}.

\subsection{Overview of T-duality and NCPT bundles}
Suppose that $A(X)$ is a $C^*$-bundle over a
locally compact space $X$ with a fibrewise action of a torus
$\TT^n$, and that $A(X) \rtimes \TT^n \cong \CT(X, H_3)$, where $\CT(X,
H_3)$ is a continuous trace algebra with spectrum $X$ and
Dixmier-Douady class $H_3 \in H^3(X; \ZZ)$. These
$C^*$-bundles are called $H_3$-twisted NCPT (noncommutative principal torus) bundles over $X$. Our first
main result there is that any $H_3$-twisted NCPT bundle $A(X)$ is
equivariantly Morita equivalent to the parametrised deformation
quantization of the continuous trace algebra $$\CT(Y_{H_2},
q^*(H_3))_{\sigma},$$ where $q: Y_{H_2} \to X$ is a principal torus
bundle with Chern class  $H_2 \in H^2(X;
H^{1}(\TT^n;\ZZ))$, and $\sigma\in C_b(X, Z^2(\ZZ^n,U(1)))$
a defining parametrised deformation such that $[\sigma] = H_1 \in H^1(X;
H^{2}(\TT^n;\ZZ))$. Here $\ZZ^n$ is the Pontryagin dual of $\TT^n$. This enables us to prove that the
continuous trace algebra $$\CT(X\times \TT^n, H_1 + H_2
+H_3),$$ with Dixmier-Douady class $H_1+H_2+H_3 \in
H^3(X \times \TT^n; \ZZ)$ where $H_j \in H^j(X; H^{3-j}(\TT^n;\ZZ))$,
has an action of the vector group $\RR^n$ that is the universal
cover of the torus $\TT^n$, and covering the $\RR^n$-action on $X \times \TT^n$.  Moreover the crossed product can be
identified up to $\TT^n$-equivariant Morita equivalence,
\beq\label{Morita}
\CT(X\times \TT^n, H_1 + H_2 +H_3) \rtimes \R^n \cong \CT(Y_{H_2}, q^*(H_3))_{\sigma}.
\eeq
That is, the T-dual of $(X\times \TT^n, H_1 + H_2 +H_3)$ is the parametrised strict deformation quantization
of $(Y_{H_2}, q^*(H_3))$ with deformation parameter $\sigma,\, [\sigma]=H_1$. From this we obtain the explicit
dependence of the $K$-theory of $\CT(Y_{H_2}, q^*(H_3))_{\sigma}$ in terms of the deformation parameter.

In particular, by equation \eqref{Morita} and the Connes--Thom isomorphism \cite{Connes81}, we deduce the isomorphism
\beq
K^j(X\times \TT^n, H_1 + H_2 +H_3) \cong K_{j+n}(\CT(Y_{H_2}, q^*(H_3))_{\sigma}).\label{Tdualityisomorphism}
\eeq
This is a more explicit version of a special case of noncommutative T-duality considered in \cite{MR05,MR06}.

\subsection{Boundary NCPT bundles and the bulk-boundary map}\label{section:bulkboundaryoverview}
In this paper, we will regard $\CT(Y_{H_2},q^*(H_3))_{\sigma}$ as the \emph{bulk NCPT bundle} over $X$. From the above discussion, we have $$\CT(Y_{H_2},q^*(H_3))_{\sigma}\rtimes \TT^n \sim \CT(X, H_3),$$ so Takai duality yields a Morita equivalence
\beq
\CT(Y_{H_2},q^*(H_3))_{\sigma}\sim \CT(X, H_3)\rtimes\ZZ^n.\label{continuoustraceidentity}
\eeq
We wish to ``peel off'' the action of a particular $\ZZ$ subgroup which we denote by $\bZ=\widehat{\bT}$, where $\bT\subset \TT^n= \TT^{n-1}\times\TT\equiv\aT^{n-1}\times\bT=\widehat{\aZ}^{n-1}\times\widehat{\bZ}$. That is, $$\CT(X, H_3)\rtimes\ZZ^n\cong (\CT(X, H_3)\rtimes\aZ^{n-1})\rtimes\bZ.$$

Under the inclusion $\iota:X\times \aT^{n-1}\rightarrow X\times\TT^n$, we have the restriction  
$$\iota^*(X\times \TT^n,\, H_1 + H_2 +H_3) = (X\times \aT^{n-1},\, \iota^*H_1 + \iota^*H_2 +H_3).$$
As in \eqref{Tdualityisomorphism}, the restricted T-dual transformation (as noncommutative principal $\TT^{n-1}$-bundles) is also an isomorphism in $K$-theory,
\beq
K^j(X\times \aT^{n-1},\, \iota^*H_1 + \iota^*H_2 +H_3) \overset{T_a}{\cong} K_{j+n-1}\left(\CT(Y_{\iota^*H_2}, q_a^*(H_3))_{\iota^*\sigma}\right).\label{restrictedTdualityisomorphism}
\eeq
On the right-hand-side of \eqref{restrictedTdualityisomorphism}, $Y_{\iota^*H_2}$ is a principal $\aT^{n-1}$-bundle over $X$ whose Chern class is equal to 
$\iota^*H_2 \in H^2(X;H^{1}(\aT^{n-1};\ZZ))$, its bundle projection is $q_a:Y_{\iota^*H_2}\rightarrow X$, and $[\iota^*\sigma] = \iota^*H_1 \in H^1(X;
H^{2}(\TT^{n-1};\ZZ))$ is the induced parametrised deformation. Note that the $\TT^n$-bundle $Y_{H_2}$ can also be regarded as a principal $\bT$-bundle over $Y_{\iota^*H_2}$, with bundle projection $q_b$ being the quotient under the action of the subgroup $\bT\subset\TT^n$; thus $q=q_a\circ q_b$. We will regard $\CT(Y_{\iota^*H_2},\, q_a^*(H_3))_{\iota^*\sigma}$ as the \emph{boundary NCPT bundle} over $X$.

As in \eqref{continuoustraceidentity}, we also have a Morita equivalence $\CT(Y_{\iota^*H_2},\, q_a^*(H_3))_{\iota^*\sigma}\sim \CT(X, H_3)\rtimes\aZ^{n-1}$. Taking crossed products with $\bZ$, we obtain
\beq
\CT(Y_{\iota^*H_2},\, q_a^*(H_3))_{\iota^*\sigma}\rtimes\bZ\sim\CT(X, H_3)\rtimes\ZZ^n\sim \CT(Y_{H_2},\, q^*(H_3))_{\sigma},\nonumber
\eeq
exhibiting the bulk NCPT bundle as a crossed product of the boundary NCPT bundle by $\bZ$. Associated to this $\bZ$ action is the Pimsner--Voiculescu \emph{bulk-boundary map} \cite{Pimsner} 
$$\partial:K_{j+n}(\CT(Y_{H_2},\, q^*(H_3))_{\sigma})\rightarrow K_{j+n-1}(\CT(Y_{\iota^*H_2},\, q_a^*(H_3))_{\iota^*\sigma}).$$

In general, we expect the following to hold.

\begin{conjecture}\label{generalconjecture}
The following diagram commutes,
\beq\label{BB1}
\xymatrix{
K^j(X\times \TT^n,\, H_1 + H_2 +H_3)  \ar[d]^{\iota^*} \ar[rr]^{\sim\;}_{T\;} && K_{j+n}\left(\CT(Y_{H_2},\, q^*(H_3))_{\sigma}\right) \ar[d]^\partial \\
K^j(X\times \TT^{n-1},\, \iota^*H_1+\iota^*H_2 +H_3) \ar[rr]^{\sim}_{T_a} && K_{j+n-1}\left(\CT(Y_{\iota^*H_2},\, q_a^*(H_3))_{\iota^*\sigma}\right)  }  
\eeq
which will show that the bulk-boundary correspondence is trivialised by T-duality in this parametrised context. Here, $T_a$ is noncommutative T-duality with respect to $\TT^{n-1}=\aT^{n-1}$, and $\iota^*$ is the induced restriction map in (twisted) $K$-theory under the inclusion $\iota:X\times \TT^{n-1}\rightarrow X\times\TT^n$.
\end{conjecture}

For a general reference on $C^*$-algebras and $K$-theory, see \cite{Blackadar},  for a general reference on $C^*$-crossed product algebras see \cite{Williams}, for a general reference on continuous trace $C^*$-algebras, see \cite{Raeburn-Williams} and for their $K$-theory, see \cite{Rosenberg3}.  A general reference on Algebraic Topology is \cite{Hatcher}.

\subsection{Ordinary Fourier transform of restriction is integration}
T-duality can be thought of as a generalised Fourier transform which, instead of transforming ordinary functions, gives an isomorphism at the level of topological invariants. Consider the Fourier transform $\mathrm{FT}_{\TT^d}:f\mapsto\widehat{f}$ which takes $f:\ZZ^d\rightarrow\CC$ to $\widehat{f}:\widehat{\ZZ^d}=\TT^d\rightarrow\CC$. This is implemented by the kernel $P({\bf n},{\bf k})=e^{i {\bf n}\cdot {\bf k}},\, {\bf n}\in\ZZ^d, {\bf k}\in\TT^d$,
\beq
    \widehat{f}({\bf k})=\sum_{\bf n} P({\bf n},{\bf k})f({\bf n})=\sum_{\bf n} e^{i {\bf n}\cdot {\bf k}}f({\bf n})\nonumber,
\eeq
while $P({\bf n},{\bf k})^{-1}$ implements the inverse transform with a similar formula (the Chern character for the \emph{Poincar\'{e} line bundle $\mathcal{P}\rightarrow\TT^d\times\TT^d$} is the analogous object in the Fourier--Mukai transform in T-duality). 

Write $({\bf m},n_d)={\bf n}$ and let $\iota$ be the inclusion of $\ZZ^{d-1}\rightarrow\ZZ^d$ taking ${\bf m}\mapsto({\bf m},0)$. Let $\partial:\widehat{f}\mapsto \partial\widehat{f}$ be partial integration along the $d$-th circle in $\TT^d$. Since this picks out only the part of $\widehat{f}$ with Fourier coefficient $n_d=0$, it follows that there is a commutative diagram
\beq\label{FTdiagram}
\xymatrix{
f  \ar[d]^{\iota^*} \ar[rr]^{\sim\;}_{\mathrm{FT}_{\TT^d}} && \widehat{f} \ar[d]^\partial \\
\iota^*f \ar[rr]^{\sim}_{\mathrm{FT}_{\TT^{d-1}}} && \partial\widehat{f}}.
\eeq
If we view $C(\TT^d)$ as a crossed product $C(\TT^{d-1})\rtimes\dZ$ with trivial action of the $d$-th copy of $\ZZ$, and represent the torus $K$-theory classes by differential forms, the Pimsner--Voiculescu boundary map is implemented by integration (or push-forward) along $\dT=\widehat\dZ$. The general conjecture \ref{generalconjecture} is in this way suggested by the elementary diagram \eqref{FTdiagram}.

\section{T-duality trivialises bulk-boundary correspondence in parametrised context}

\label{section:Tdualitytrivial}
In this paper, we will restrict ourselves to the special case where $X=\TT=S^1$, $n=2$, $H_3=0=H_2$ and $[\sigma^1]=H_1\in H^1(S^1, H^2(\TT^2,\ZZ))\cong\ZZ$ the decomposable generator. When regarded as an element of $H^3(S^1\times\TT^2;\ZZ)$, $H_1$ is the volume form of $S^1\times\TT^2\cong \TT^3$. Its corresponding deformation parameter $\sigma^1\in C(S^1,Z^2(\ZZ^2,U(1)))\cong C(S^1, \TT)$ is the identity function; more explicitly, the $U(1)$-valued multiplier on $\ZZ^2\times\ZZ^2$ at the point $e^{2\pi i\theta}\in S^1$ is 
\begin{equation}
\sigma^1(e^{2\pi i\theta}):((p,q),(r,s))\mapsto e^{i\theta (ps-qr)},\qquad (p,q),(r,s)\in\ZZ^2.\label{basicdeformationparameter}
\end{equation}
For $kH_1$, the deformation parameter is $\sigma^k$, i.e.\ $kH_1=k[\sigma^1]=[\sigma^k]$. 

In these cases, the top row of Conjecture \ref{generalconjecture} is
\beq
K^j(S^1\times \TT^2, k H_1) \overset{T}\cong K_{j}(C(S^1 \times \TT^2)_{\sigma^k}).\label{specialcasetoprow}
\eeq
The general case is work in progress.

\subsection{Generalities on the integer Heisenberg group}\label{section:Heisenberggeneralities}
The right-hand-side of \eqref{specialcasetoprow} (the bulk NCPT bundle) is the $C^*$-algebra $C(S^1 \times \TT^2)_{\sigma^k}$, which is the parametrised deformation quantization of $C(S^1 \times \TT^2)$ by $\sigma^k$. By section 5 \cite{MR05} we see that for $k\ne 0$, 
\beq
C(S^1 \times \TT^2)_{\sigma^k} \cong C^*(\heisZ(k))\nonumber
\eeq
where the integer Heisenberg group $\heisZ(k)$ is defined as
\beq
\heisZ(k) = \left\{\left(\begin{array}{ccc}1 & a & \frac{c}{k} \\0 & 1 & b \\0 & 0 & 1\end{array}\right) \Big|\, a,b,c \in \ZZ\right\}.\nonumber
\eeq
We also write $\heisZ\coloneqq\heisZ(1)$.

As we will need to keep track of the various subgroups in $\heisZ$, we introduce the following notation. The normal subgroup of matrices with $a=0$ (resp.\ $b=0$) is denoted by $\bcZ$ (resp.\ $\acZ$). The subset with $c=0$ is denoted by $\abZ$. The central subgroup with $a=0=b$ is denoted by $\cZ$, while the subgroup with $b=0=c$ (resp.\ $a=0=c$) is denoted by $\aZ$ (resp.\ $\bZ$). Thus $\heisZ(k)$ is a non-split central extension of $\ZZ^2$ by $\ZZ$,
\begin{equation}
0\longrightarrow \cZ\longrightarrow \heisZ(k)\longrightarrow\abZ\longrightarrow 0,\label{heisenbergasextension}
\end{equation}
where we have reused the symbol $\abZ$ for the quotient group. We will peel off the action of $\bZ$ as described in Section \ref{section:bulkboundaryoverview}

If we label elements of $\heisZ(k)$ by the 3-tuple $(a,b,c)\in\ZZ^3$, the group multiplication is $(a_1,b_1,c_1)\cdot(a_2,b_2,c_2)=(a_1+a_2, b_1+b_2, c_1+c_2+ka_1b_2)$. Then the group 2-cocycle for the central extension \eqref{heisenbergasextension} is $\sigma_k^{\mathrm{group}}((a_1,b_1),(a_2,b_2))=ka_1b_2$.

We can also write $\heisZ(k)$ as a semi-direct product in two ways,
\beq\label{eqn:semiz}
\heisZ(k)\cong \bcZ\rtimes\aZ\cong \acZ\rtimes\bZ.
\eeq
For example, the action of $b\in\bZ$ on $\acZ$ by conjugation takes $(a,c)\mapsto(a,c-kba)$. 
 The action of $b\in\bZ$  on $\acZ$ can also be expressed via an ${\rm SL}(2, \ZZ)$ matrix as follows,
\beq\label{eqn:sl(2z)}
(a,c)\left(\begin{array}{cc}1 & -kb \\0 & 1\end{array}\right) = (a,c-kba).
\eeq
In particular, using \eqref{eqn:semiz}, we get the (split) short exact sequence of groups,
\beq\label{heiszextn2}
0\longrightarrow  \acZ \longrightarrow \heisZ(k) \longrightarrow \bZ \longrightarrow 0.
\eeq
The Pontryagin dual of $\acZ$ will be denoted by $\acT$, and similarly for the other subgroups. For instance, we have $C^*(\acZ)\cong C(\acT)$.

The group $C^*$-algebra $C^*(\heisZ(k))$ is generated by three unitaries $U,V,W$ subject to the relation $UV=W^kVU$ and $W$ being central. We can view $U,V,W$ as the respective images in $C^*(\heisZ(k))$ of the generators of $\aZ, \bZ, \cZ$. As $\heisZ(k)$ is a semi-direct product, we may write $C^*(\heisZ(k))$ as a crossed product 
\beq
C^*(\heisZ)\cong C^*(\acZ)\rtimes\bZ\cong C(\acT)\rtimes\bZ,
\eeq
with $C(\acT)$ the boundary NCPT bundle. It is convenient to think of $\aT,\bT,\cT$ as unit circles in the complex plane, whose points are respectively labelled by the complex numbers $u,v,w$, then $U,V,W$ are the identity functions on these unit circles. Then the generating automorphism $\alpha_k$ in the crossed product $C(\acT)\rtimes\bZ\equiv C(\acT)\rtimes_{\alpha_k} \bZ$ acts on $f\in C(\acT)$ by 
\begin{equation}
(\alpha_k\cdot f)(u,w)=f(w^{-k}u,w).\label{actiondefiningheisenberg}
\end{equation}
Alternatively, it acts on the unitaries $U^aW^c\in C(\acT)$ by 
\begin{equation}
U^aW^c\mapsto U^aW^{c-ka}.\label{actiondefiningheisenberg2}
\end{equation}

From the parametrised viewpoint, $C^*(\heisZ(k))$ is also a twisted crossed product $$C^*(\heisZ(k))\cong C(\cT)\rtimes_{\sigma^k}\abZ$$ with $U(C(\cT))$-valued cocycle $\sigma^k$ (c.f.\ Eq.\  \eqref{basicdeformationparameter}). Thus, we can regard $C^*(\heisZ(k))$ as a continuous field of noncommutative tori parametrised by $S^1=\cT$, with the rotation angle of the noncommutative torus over a point $w=e^{2\pi i \theta}\in S^1$ being $2\pi k\theta$. To emphasize this parametric point of view, we will often write $S^1$ in place of $\cT$.

\subsection{Statement of main results}

Associated to the $\bZ$ action on $C(\acT)$ is the Pimsner--Voiculescu boundary map $$\partial:K_0(C^*(\heisZ(k)))\rightarrow K_1(C(\acT))\cong K^1(S^1\times\aT).$$

The first main result in our paper is, 

\begin{theorem}\label{thm:main}
The following diagram commutes,
\beq\label{BB}
\xymatrix{
K^0(S^1\times \TT^2,\, k H_1)  \ar[d]^{\iota^*} \ar[r]^{\sim\;}_{T\;} & K_0(C^*(\heisZ(k))) \ar[d]^\partial \\
K^0(S^1 \times \TT) \ar[r]^{\sim}_{T_a} & K^1(S^1 \times \TT)  }  
\eeq
which shows that the bulk-boundary correspondence is trivialised by T-duality in this parametrised context. Here, $T_a$ is T-duality (Fourier--Mukai transform) with respect to $\TT=\aT$, and $\iota^*$ is the induced restriction map in (twisted) $K$-theory under the inclusion $\iota:S^1\times \TT\rightarrow S^1\times\TT^2$.
\end{theorem}

Theorem \ref{thm:main} is a special case of Conjecture \ref{generalconjecture}. We will prove the commutation of the diagram \eqref{BB} by first writing $T$ as a composition of a $T$-duality isomorphism $T_1$ for circle bundles, followed by a Baum--Connes isomorphism \cite{BCH} $T_2$, which in this case we show is a consequence of an imprimitivity theorem \cite{Rieffelsymmetric} and the Connes--Thom
isomorphism theorem \cite{Connes81} (Proposition \ref{factorization}). Then we compute the effect of the various maps on explicit $K$-theory generators. The proof of Theorem \ref{thm:main} is assembled in Section \ref{section:mainproof}. Along the way, we also prove a second related result:
\begin{theorem}\label{thm:second}
The following diagram commutes,
\beq\label{BB2}
\xymatrix{
K^1(\nil_k)  \ar[d]^{\iota^*} \ar[r]^{\sim\qquad}_{T_2\qquad} & K_0(C^*(\heisZ(k))) \ar[d]^\partial \\
K^1(S^1 \times \TT) \ar[r]^{\sim}_{T_{ac}} & K^1(S^1 \times \TT)  }
\eeq
where $\nil_k$ is the Heisenberg nilmanifold, $T_2$ is the Baum--Connes isomorphism described in the previous paragraph, $T_{ac}$ is the full T-duality isomorphism (Fourier--Mukai transform) with respect to $S^1\times\TT=\TT^2=\acT$, and $\iota$ is a fibre inclusion. Here $\nil_k$ is regarded as a fibre bundle over the circle $\bT$ with typical fibre $\acT$.
\end{theorem}
As we shall see, Theorems \ref{thm:main} and \ref{thm:second} can actually be summarised by a single diagram \eqref{BBfull}.

Let us elaborate on the description of $\nil_k$ as a bundle over $\bT$ in Theorem \ref{thm:second}. The real Heisenberg group $\heisR$ is defined as
\beq
\heisR = \left\{\left(\begin{array}{ccc}1 & a & c\\0 & 1 & b \\0 & 0 & 1\end{array}\right) \Big|\, a,b,c \in \RR\right\}\nonumber
\eeq
and $\heisZ(k)$ sits inside $\heisR$ as a discrete subgroup. The Heisenberg nilmanifold is $$\nil_k=\heisR/\heisZ(k).$$ We can write $\heisR$ as a semi-direct product in two ways,
\beq\label{eqn:semir}
\heisR\cong \bcR\rtimes\aR\cong \acR\rtimes\bR.
\eeq
For the second semi-direct product, the action of $b\in\bR$ on $\acR$ by conjugation takes $(a,c)\mapsto(a,c-ba)$. 
As in \eqref{eqn:sl(2z)real}, the action of $b\in\bR$  on $\acR$ can be expressed via an ${\rm SL}(2, \RR)$ matrix as follows,
\beq\label{eqn:sl(2z)real}
(a,c)\left(\begin{array}{cc}1 & -b \\0 & 1\end{array}\right) = (a,c-ba).
\eeq
Note that the restriction of \eqref{eqn:sl(2z)real} to an action of $b\in\bZ$ on the discrete subgroup $\acZ$ is $(a,\frac{c}{k})\mapsto(a,\frac{c-kba}{k})$. Equivalently, this action takes $(a,c)\mapsto(a,c-kba)$ after relabelling the elements of $\acZ$ and $\heisZ(k)$ with integers $a,b,c$ according to our convention in Section \ref{section:Heisenberggeneralities}. Thus \eqref{eqn:sl(2z)real} subsumes \eqref{eqn:sl(2z)}.

From \eqref{eqn:semir}, we get the (split) short exact sequence of groups,
\beq\label{heisrextn2}
0\longrightarrow  \acR \longrightarrow \heisR \longrightarrow \bR \longrightarrow 0.
\eeq
Using \eqref{heisrextn2} and \eqref{heiszextn2} and taking the quotient, we arrive at the description of $\nil_k$ regarded as a (non-principal) torus fibre bundle over the circle $\bT=\bR/\bZ$, or mapping torus, with fibre $F\sim\acT=\acR/\acZ$,
\beq
\acT \longrightarrow \nil_k \longrightarrow \bT.\label{eqn:nilastorusbundle}
\eeq
More explicitly, the ${\rm SL}(2,\ZZ)$ transformation $\begin{pmatrix}1 & -k \\ 0 & 1\end{pmatrix}$ gives the diffeomorphism of the fibre specifying the bundle, and the map $\iota$ is taken to be a fibre inclusion.

\subsection{Factorization of T-duality}
\label{factorization}
The Heisenberg nilmanifold $\nil_k=\heisR/\heisZ(k)$ has a left action of the real Heisenberg group $\heisR$. Since $\heisR$ is a central extension,
\beq
0\to \RR \to \heisR \to \RR^2\to 0,\nonumber
\eeq
we can write it as a twisted Cartesian product $\heisR = \RR \times_\omega \RR^2$, where $\omega$ is the standard symplectic form on the vector space $\RR^2$
and twists the product on $\RR\times \RR^2$ in the usual way. 

We have, on the one hand, 
\begin{align}
\label{eqn:me}
C(\nil_k) \rtimes \heisR = (C(\nil_k) \rtimes \RR) \rtimes_\omega \RR^2 & \sim \CT(\TT^3, kH_1) \rtimes_\omega \RR^2\nonumber \\
\CT(\TT^3, kH_1) \rtimes_\omega \RR^2\otimes \cK & \cong \CT(\TT^3, kH_1) \rtimes\RR^2
\end{align}
where $\sim$ refers to Morita equivalence and the isomorphism in the second line is a consequence of the Packer-Raeburn trick,  \cite{PR}. The 3-torus $\TT^3=S^1\times \abT$ can be regarded as a circle bundle over $\abT$ with $H$-flux $kH_1$, and we observe that T-duality for this circle bundle \cite{BEM} is
\beq\label{eqn:T_1}
 K_\bullet(C(\nil_k) \rtimes \RR) \stackrel{\text{Connes--Thom}}{\cong} K^{\bullet+1} (\nil_k)  \overset{T_1^{-1}}{\cong}  K_\bullet(\CT(\TT^3, kH_1)).
\eeq
Note that $\nil_k$ is regarded here as a principal circle bundle over $\abT$ (se Section \ref{nil}). Also by noncommutative T-duality for $\TT^3$ as a torus bundle over $S^1$, Section 5 \cite{MR05},
\beq\label{eqn:T}
 K_\bullet(\CT(\TT^3, kH_1) \rtimes \RR^2) \stackrel{\text{Connes--Thom}}{\cong} K^{\bullet} (\TT^3, kH_1)  \stackrel{T}{\cong}  K_\bullet(C^*(\heisZ(k))).
\eeq

On the other hand, by the symmetric imprimitivity Theorem \cite{Rieffelsymmetric}, the 
crossed product $C^*$-algebra $C(\nil_k) \rtimes \heisR$ is strongly Morita equivalent to $C(\heisR\backslash \heisR)\rtimes \heisZ(k) = C^*(\heisZ(k))$, and in particular, 
\beq
K_\bullet(C(\nil_k) \rtimes \heisR) \cong K_\bullet(C^*(\heisZ(k))).\label{GreenHeisenberg}
\eeq
By the Connes--Thom isomorphism Theorem (c.f.\ Corollary 7 of \cite{Connes81}, Corollary 2 of \cite{Fack81}), there is a natural isomorphism
\beq
K_\bullet(C(\nil_k) \rtimes \heisR) \cong K_{\bullet+1}(C(\nil_k)).\label{ConnesThomHeisenberg}
\eeq
From \eqref{GreenHeisenberg} and \eqref{ConnesThomHeisenberg}, we conclude that 
\beq\label{eqn:T_2}
K^{\bullet+1}(\nil_k) =K_{\bullet+1}(C(\nil_k)) \stackrel{T_2}{\cong} K_\bullet(C^*(\heisZ(k))),
\eeq
which can also be understood as the the assembly map \cite{BCH} composed with Poincar\'e duality in this context.
Equations \eqref{eqn:me}, \eqref{eqn:T_1}, \eqref{eqn:T}, \eqref{eqn:T_2} imply the claimed factorization of T-duality, 

\begin{proposition}[Factorization of T-duality]\label{prop:factorization}
In the notation above, the following diagram commutes,

\beq\label{eqn:factorization}
\xymatrix{K^{\bullet} (\TT^3, kH_1) \ar[ddr]_{T_1}^{\sim} \ar[rr]_{T}^{\sim}&&
K_\bullet(C^*(\heisZ(k)))\\
&&\\
&K^{\bullet+1} (\nil_k) \ar[uur]_{T_2}^{\sim} &}
\eeq
\end{proposition}

For the 2-torus $S^1\times\TT=\cT\times\aT$, T-duality $T_{ac}$ is just the ordinary Fourier--Mukai transform. It factorises as (the inverse\footnote{The Fourier transform and its inverse are the same up to a sign choice in the exponent of the kernel $P({\bf n},{\bf k})=e^{\pm i{\bf n}\cdot{\bf k}}$. Similarly the Fourier--Mukai transform and the inverse are the same up to a sign choice for first Chern class of the Poincar\'{e} line bundle.} of) T-duality $T_c$ with respect to one circle factor $S^1=\cT$, followed by T-duality $T_a$ with respect to the other circle factor $\aT$,
\beq\label{eqn:factorization2}
\xymatrix{&K^{\bullet} (S^1\times\TT)  \ar[ddr]_{T_{ac}}^{\sim} &\\
&&\\
K^{\bullet+1} (S^1\times\TT) \ar[uur]_{T_c}^
{\sim} \ar[rr]_{T_{a}}^{\sim}&&
K^\bullet(S^1\times\TT)
}
\eeq

We use the factorisations in \eqref{eqn:factorization} and \eqref{eqn:factorization2} to rewrite and combine the commutative diagrams \eqref{BB} and \eqref{BB2} as 
\beq\label{BBfull}
\xymatrix{
K^0(S^1\times \TT^2,\, k H_1)  \ar[dr]_{T_1}^{\sim} \ar[dddd]^{\iota^*} \ar[rr]^{\sim\;}_{T\;} & &  K_0(C^*(\heisZ(k))) \ar[dddd]^\partial \\
&K^{1} (\nil_k) \ar[ur]_{T_2}^{\sim} \ar[dd]^{\iota^*} &\\
&&\\
&K^1(S^1\times\TT)\ar[dr]_{T_{ac}}^{\sim} &\\
K^0(S^1 \times \TT) \ar[rr]^{\sim}_{T_a} \ar[ur]_{T_c}^{\sim} &&  K^1(S^1 \times \TT)  }  
\eeq

\section{Geometry of the Heisenberg Nilmanifold, $K$-theory of integer Heisenberg group algebra and PV sequence}

\subsection{$K$-theory generators}

\subsubsection{Generators for $K^0(S^1\times \TT^2, k H_1)$}
Let us coordinatise $S^1\eqqcolon\cT$ by $x_c$ and $\TT^2\equiv\abT=\aT\times\bT$ by $(x_a, x_b)$. By equation (4.15) \cite{BEM,BEM2} (see also section 4.1.4 \cite{BS}), we see that there is a natural isomorphism when $k\ne 0$,
\begin{align}
K^0(S^1\times \TT^2, k H_1) & \cong \widetilde{K}^0(S^1\times \TT^2) \cong  \ZZ[\cL_{ca}]  \oplus \ZZ[\cL_{cb}] \oplus \ZZ[\cL_{ab}] \nonumber \\
& \cong H^2(S^1\times \TT^2; \ZZ) \cong \ZZ[dx_c\wedge dx_a]  \oplus \ZZ[dx_c\wedge dx_b] \oplus \ZZ[dx_a\wedge dx_b] \nonumber
\end{align}
where $\widetilde{K}^0(S^1\times \TT^2)$ denotes the reduced $K$-theory, while for $i\neq j\in\{a,b,c\}$, $\cL_{ij}$ denotes the line bundle supported on 
the subtorus $\underset{i}{\TT}\times\underset{j}{\TT}$ of $S^1\times \TT^2$ with first Chern class equal to $[dx_i\wedge dx_j]$. 

Note that unlike in the untwisted ($k=0$) case, there is no generator for the ``rank'' in twisted $K$-theory. Strictly speaking, $\cL_{ij}$ should refer to the rank-0 virtual bundle $\wt{\cL}_{ij}\coloneqq \cL_{ij}-{\bf 1}$. We will write the latter when this clarification is required.

\subsubsection{Generators for $K^\bullet(S^1\times\TT^2)$}
The $K^0$ group of the 2-torus $S^1\times\TT=\cT\times \aT$ is simply
\begin{equation}
K_0(C(S^1\times\TT))=\ZZ[{\bf 1}]\oplus \ZZ[P_\mathrm{Bott}-{\bf 1}]=\ZZ[{\bf 1}]\oplus\ZZ[dx_c\wedge dx_a],\nonumber
\end{equation}
where we have used the Chern character isomorphism in the second equality to represent the second generator in terms of a differential form. Similarly, the $K^1$ group is
\begin{equation}
K_1(C(S^1\times\TT))\cong H^1(S^1\times\TT)=\ZZ[U_c]\oplus\ZZ[U_a]=\ZZ[dx_c]\oplus \ZZ[dx_a],\nonumber
\end{equation}
where $U_c, U_a : S^1\times \TT \to U(1)$ are 
continuous maps that are also generators of $H^1(S^1\times \TT;\ZZ)$ with odd degree 1 Chern class $[dx_c]$ and $[dx_a]$ respectively.

\subsubsection{Generators for $K_0(C^*(\heisZ(k)))$}\label{section:Heisenbergalgebraktheory}
The Bott projections for the commutative subalgebras $C^*(U,W)\cong C(\acT)$ and $C^*(V,W)=C^*(\bcT)$ are denoted by $[P_{ac}]$ and $[P_{bc}]$ respectively (see e.g.\  \cite{AP} for explicit formulae).

By results of \cite{AP, Ko} (see also the derivation of \eqref{heisenbergPVsequence} in Section \ref{subsection:PVboundary}),  $K_0(C^*(\heisZ(k)))$ is generated by the free module of rank one, $[{\bf 1}]$,
and also by the projections $[P_{ac}]$ and $[P_{bc}]$; that is,
\beq
K_0(C^*(\heisZ(k))) \cong \ZZ[{\bf 1}] \oplus \ZZ[P_{ac}] \oplus \ZZ[P_{bc}].\nonumber
\eeq 
Note that there is no ``Bott projection'' $P_{ab}$ in $C^*(\heisZ(k))$ since $C^*(U,V)$ is not a subalgebra of $C^*(\heisZ(k))$.

\subsubsection{$K$-theory of Heisenberg nilmanifold}\label{nil}
The Heisenberg nilmanifold 
$$\nil_k=\heisR/\heisZ(k)$$ is a principal circle bundle over $\abT$ with Chern class equal to $k[dx_a\wedge dx_b]$, and is also a classifying space for $\heisZ(k)$.

The generators for $K^1(\nil_k)$ are described as follows:
\begin{align}
K^1(\nil_k)& \cong   \ZZ[Y]  \oplus \ZZ[U_a] \oplus \ZZ[U_b] \nonumber\\
 \cong H^{odd}(\nil_k) & \cong \ZZ[dx_a\wedge dx_b\wedge \widehat A]  \oplus \ZZ[dx_a] \oplus \ZZ[dx_b] \nonumber
\end{align}
where $\widehat A$ is a connection on the principal circle bundle $\nil_k$ over $\abT$ with curvature $d\widehat A = k dx_a\wedge dx_b$. Here $Y: \nil_k \to SU(2)$ is a degree 1 continuous map
with odd degree-3 Chern class $[dx_a\wedge dx_b\wedge \widehat A]$ which is the generator of $H^3(\nil_k;\ZZ)$, 
while $U_a, U_b : \nil_k \to U(1)$ are 
continuous maps that are also generators of $H^1(\nil_k;\ZZ)$ with odd degree-1 Chern classes $[dx_a]$ and $[dx_b]$ respectively.

\section{Proof of main results}

\subsection{Explicit maps on generators}

\subsubsection{Bulk-boundary map $\partial$}\label{subsection:PVboundary}
Recall from \eqref{actiondefiningheisenberg} that $C^*(\heisZ(k))\cong C(S^1\times\aT)\rtimes_{\alpha_k}\bZ$. Its $K$-theory groups can be calculated along the lines of \cite{AP} (which handled the $k=1$ case). The PV sequence (\cite{Pimsner}) associated to the action $\alpha_k$ of $\bZ$ on $C(S^1\times\aT)$ is
\begin{equation}
0\rightarrow {\rm coker}_{1-\alpha_{k*}}(K_\bullet(C(S^1\times\aT))\rightarrow K_\bullet(C^*(\heisZ(k)))\xrightarrow{\partial} {\rm ker}_{1-\alpha_{k*}}(K_{\bullet-1}(C(S^1\times\aT))\rightarrow 0\nonumber
\end{equation}

Using the explicit expression \eqref{actiondefiningheisenberg}, we can compute the induced map $\alpha_{k*}$ on $K$-theory as follows. It can be shown that $\alpha_{k*}$ acts trivially on $K_0(C(S^1\times\aT))=\ZZ[{\bf 1}]\oplus\ZZ[P_{ac}]$. Thus
\begin{align}
{\rm ker}_{1-\alpha_{k*}}(K_0(C(S^1\times\aT)) & = \ZZ[{\bf 1}]\oplus \ZZ[P_{ac}] ,\nonumber\\
{\rm Im}_{1-\alpha_{k*}}(K_0(C(S^1\times\aT)) & = [{\bf 0}] ,\nonumber\\
{\rm coker}_{1-\alpha_{k*}}(K_0(C(S^1\times\aT)) & = \ZZ[{\bf 1}]\oplus \ZZ[P_{ac}].\nonumber
\end{align}

Next, $K^1(S^1\times\aT)\cong K_1(C(S^1\times\aT))=\ZZ[W]\oplus\ZZ[U]$, and we recall from \eqref{actiondefiningheisenberg2} that $\alpha_k(U^aW^c)=U^aW^{c-ka}$. Therefore,
\begin{align}
{\rm ker}_{1-\alpha_{k*}}(K_1(C(S^1\times\aT)) & = \ZZ[W] ,\nonumber\\
{\rm Im}_{1-\alpha_{k*}}(K_1(C(S^1\times\aT)) & = \ZZ[W^k] ,\nonumber\\
{\rm coker}_{1-\alpha_{k*}}(K_1(C(S^1\times\aT)) & = \ZZ_k[W]\oplus\ZZ[U].\nonumber
\end{align}

Thus, for $\bullet=0$, the PV sequence is
\begin{equation}
0\rightarrow \ZZ[{\bf 1}]\oplus\ZZ[P_{ac}]\rightarrow \ZZ[{\bf 1}]\oplus\ZZ[P_{ac}]\oplus \ZZ[P_{bc}]\xrightarrow{\partial} \ZZ[W]\rightarrow 0,\label{heisenbergPVsequence}
\end{equation}
with $\partial$ taking $[P_{bc}]\mapsto -[W]$ and annihilating $[\bf 1]$ and $[P_{ac}]$. Equivalently, with $$[\wt{P}_{ic}]\coloneqq [P_{ic}]-[{\bf 1}],\quad i=a,b,$$ we can write
\begin{equation}
\partial([\wt{P}_{bc}])=-[W],\qquad\partial([\wt{P}_{ac}])=0=\partial([{\bf 1}]).\nonumber
\end{equation}

For $\bullet=1$, the PV sequence (\cite{Pimsner})  is 
\begin{equation}
0\rightarrow \ZZ_k[W]\oplus\ZZ[U]\rightarrow \ZZ_k[W]\oplus\ZZ[U]\oplus \ZZ[V]\oplus\ZZ[V_a] \xrightarrow{\partial} \ZZ[1]\oplus\ZZ[P_{ac}]\rightarrow 0,\nonumber
\end{equation}
where $V_a$ is a unitary constructed as in pp.\ 7 of Ref.\ \cite{AP} (it is roughly the analogue of the top class in $H^3(\TT^3)$). The boundary map takes $[V]\mapsto -[{\bf 1}],\, [V_a]\mapsto[P_{ac}]$, and annihilates $[W]$ and $[U]$.

\subsubsection{Circle bundle T-duality $T_1$}
Viewing $S^1\times\TT^2$ as a circle bundle over $\abT$, we use the results of \cite{BEM,BEM2} to deduce that $T_1$, the T-duality isomorphism for circle bundles, maps \mbox{$K^0(S^1\times\TT^2, kH_1)$} to $K^1(\nil_k)$ as follows:
\begin{align}
T_1([dx_c\wedge dx_a]) & = [dx_a]\nonumber\\
T_1([dx_c\wedge dx_b]) & = [dx_b]\nonumber\\
T_1([dx_a\wedge dx_b]) & = [dx_a\wedge dx_b \wedge \widehat A].\nonumber
\end{align}
From this, we deduce that on the (reduced) $K$-theory classes,
\begin{align}
T_1([\wt{\cL}_{ca}]) & = [U_a]\nonumber\\
T_1([\wt{\cL}_{cb}])  & = [U_b]\nonumber\\
T_1([\wt{\cL}_{ab}]) & = [Y],\nonumber
\end{align}
giving a complete description of the isomorphism $T_1$ on the level of generators.

\subsubsection{Baum--Connes map $T_2$}
The map $T_2$ can be identified with the Baum--Connes assembly map \cite{BCH} (together with Poincar\'{e} duality) for the discrete group $\heisZ(k)$, which is known to be an isomorphism, see Proposition 2.5 in \cite{Rosenberg2},
\beq
\mu : K^1(\nil_k) \to K_0(C^*(\heisZ(k))).\nonumber
\eeq
In fact, Rosenberg shows that this Baum--Connes map can deduced from the Baum--Connes map for the torus $\TT^2$ (which is essentially the Fourier transform) via a
PV sequence (\cite{Pimsner}, often used in our paper) and the 5-Lemma. 

Upon applying the von Neumann trace $\tau: C^*(\heisZ(k)) \to \CC$, we see that 
\beq
\tau(\mu(Y)) = \int_{\nil_k} Ch^{odd}(Y) = 1.\nonumber
\eeq
Note also that $$\tau(\mu(U_a))=0=\tau(\mu(U_b)).$$

Now there are two nontrivial line bundles $\cP_{ac}$ and $\cP_{bc}$ over $\nil_k$, which are explicitly constructed in Appendix \ref{section:linebundlesovernil}. Their respective Chern classes are $[dx_a\wedge \widehat A]$ and $[dx_b \wedge \widehat A]$, which themselves correspond to group 2-cocycles
$\sigma_{ac}$ and $\sigma_{bc}$ on the discrete group $\heisZ(k)$. These give rise to cyclic 2-cocycles $\tau_{ac}$ and $\tau_{bc}$ 
on the smooth subalgebra of $C^*(\heisZ(k))$ \cite{CM} (see Appendix \ref{section:cycliccocycles}), and hence define maps\footnote{When $k\neq 1$, there is also a nontrivial line bundle $\cP_{ab}$ over $\nil_k$ whose (integral) Chern class is $k$-torsion. It is obtained from pulling back $\nil_1\rightarrow \TT^2$ along the bundle projection $\nil_k\rightarrow\TT^2$. We do not get another cyclic 2-cocycle from this.} $K_0(C^*(\heisZ(k)))\rightarrow \ZZ$.

We then compute, using the Connes--Moscovici higher index theorem \cite{CM},
\beq \label{cyclic1}
\tau_{ac}(\mu(U_b)) = \int_{\nil_k} dx_a \wedge \widehat A \wedge Ch^{odd}(U_b) =  \int_{\nil_k} dx_a \wedge \widehat A \wedge dx_b = -1\nonumber
\eeq
and 
\beq \label{cyclic2}
\tau_{bc}(\mu(U_a)) = \int_{\nil_k} dx_b \wedge \widehat A \wedge Ch^{odd}(U_a) =  \int_{\nil_k} dx_b \wedge \widehat A \wedge dx_a = 1.\nonumber
\eeq

On the other hand, Proposition \ref{GCT} says that
\begin{align}
\tau_{ac}(P_{ac}) &= 1, \;\; \tau_{ac}(P_{bc})=0,\;\;\tau_{ac}({\bf 1})=0  \nonumber \\
\tau_{bc}(P_{bc}) &= 1, \;\; \tau_{bc}(P_{ac})=0,\;\;\tau_{bc}({\bf 1})=0 \nonumber\\ 
\tau(P_{ac}) &= \tau(P_{bc}) = \tau({\bf 1}) = 1.\nonumber
\end{align}
This implies that 
\begin{align}
[\mu(Y)] &= [{\bf 1}] \nonumber \\
[\mu(U_a)] &= [\wt{P}_{bc}]  \nonumber\\ 
[\mu(U_b)] & = - [\wt{P}_{ac}] .\nonumber
\end{align}
That is,
\begin{align}
T_2([Y]) & = [{\bf 1}]\nonumber\\
T_2([U_a])  & = [\wt{P}_{bc}]\nonumber\\
T_2([U_b]) & = -[\wt{P}_{ac}].\nonumber
\end{align}

\subsubsection{Torus bundle T-duality $T = T_2\circ T_1$}
We deduce that
\begin{align}
T([\wt{\cL}_{ca}]) & = [\wt{P}_{bc}] \nonumber\\
T([\wt{\cL}_{cb}])  & = -[\wt{P}_{ac}]\nonumber\\
T([\wt{\cL}_{ab}]) & = [{\bf 1}],\nonumber
\end{align}
giving a complete description of $T$ on the level of generators.

\subsubsection{Restriction maps $\iota^*$}
We have $K^0(S^1\times\TT)\cong \ZZ[{\bf 1}]\oplus \ZZ[\wt{\cL}_{ca}]$, and the restriction map on $K$-theory induced by $\iota:S^1\times\TT\rightarrow S^1\times\TT^2 \rightarrow$ is easy to compute: 
\begin{align}
\iota^*([\wt{\cL}_{ca}]) & =  [\wt{\cL}_{ca}]\nonumber\\
\iota^*([\wt{\cL}_{cb}]) & = [{\bf 0}]\nonumber\\
\iota^*([\wt{\cL}_{ab}]) & =  [{\bf 0}].\nonumber
\end{align}

For the other inclusion map $\iota:S^1\times\TT=\acT\rightarrow \nil_k$ (see \eqref{eqn:nilastorusbundle}), the restriction map $\iota^*$ in $K$-theory is simply
\begin{align}
\iota^*[Y]&=[{\bf 0}]\nonumber\\
\iota^*[U_a]&=[U_a]\nonumber\\
\iota^*[U_b]&=[{\bf 0}].
\end{align}

\subsubsection{Commutative T-duality maps $T_a, T_c, T_{ac}$}
The map $$T_a:K^0(S^1\times\TT)\longrightarrow K^1(S^1\times\TT)\cong K_1(C(S^1\times\TT))$$ in \eqref{BB} refers to the 1D commutative T-duality map with respect to $\TT=\aT$. It takes $$T_a:[{\bf 1}]\mapsto-[U],\qquad[\wt{\cL}_{ca}]\sim[dx_c\wedge dx_a]\mapsto-[dx_c]\sim-[W]$$ (the parametrising circle $S^1=\cT$ goes along for the ride).

Similarly, the map $$T_c:K^0(S^1\times\TT)\longrightarrow K^1(S^1\times\TT)$$ in \eqref{eqn:factorization2} refers to the 1D commutative T-duality map with respect to $S^1=\cT$. It takes $$T_c:[\wt{\cL}_{ca}]\mapsto[U_a],\qquad[{\bf 1}]\mapsto-[U_c].$$

Finally, the full T-duality map $$T_{ac}:K^1(S^1\times\TT)\longrightarrow K^1(S^1\times\TT)$$ in \eqref{eqn:factorization2} is the composition $T_a\circ T_c^{-1}$, which takes $$T_{ac}:[U_a]\mapsto-[W],\qquad[U_c]\mapsto-[U].$$

\subsection{Assembling $K$-theory maps together}\label{section:mainproof}

\begin{proof}[Proof of Theorem \ref{thm:main}: T-duality trivializes the bulk-boundary correspondence, first version]\label{mainproof}
We can now verify the commutativity of the diagram \eqref{BB}. 

\begin{align}
T_a\circ\iota^*([\wt{\cL}_{ca}]) & = T_a([\wt{\cL}_{ca}])\nonumber \\
& =-[W]\qquad \nonumber\\
& = \partial([\wt{P}_{bc}])\nonumber\\
& = \partial\circ T([\wt{\cL}_{ca}]),\nonumber
\end{align}
and 
\begin{align}
T_a\circ\iota^*([\wt{\cL}_{cb}]) & = [{\bf 0}] = \partial[\wt{P}_{ac}]=\partial\circ T([\wt{\cL}_{cb}])\nonumber\\
T_a\circ\iota^*([\wt{\cL}_{ab}]) & = [{\bf 0}] = \partial[{\bf 1}]=\partial\circ T([\wt{\cL}_{ab}]),\nonumber
\end{align}
so $T_a\circ\iota^*=\partial\circ T$.
\end{proof}

\begin{proof}[Proof of Theorem \ref{thm:second}: T-duality trivializes the bulk-boundary correspondence, second version]\label{secondproof}
We can also verify the commutativity of the diagram \eqref{BB2}. 
\begin{align}
T_{ac}\circ\iota^*([U_a]) & = T_{ac}([U_a])\nonumber \\
& = -[W]\qquad \nonumber\\
& = \partial([\wt{P}_{bc}])\nonumber\\
& = \partial\circ T_2([U_a]),\nonumber
\end{align}
and 
\begin{align}
T_{ac}\circ\iota^*([U_b]) & = [{\bf 0}] = \partial(-[\wt{P}_{ac}])=\partial\circ T_2([U_b])\nonumber\\
T_{ac}\circ\iota^*([Y]) & = [{\bf 0}] = \partial[{\bf 1}]=\partial\circ T_2([Y]),\nonumber
\end{align}
so $T_{ac}\circ\iota^*=\partial\circ T_2$.
\end{proof}
Alternatively, Theorem \ref{thm:second} follows directly from Theorem \ref{thm:main}, the easily verified commutativity of 
\beq
\xymatrix{
K^0(S^1\times \TT^2,\, k H_1)  \ar[d]^{\iota^*} \ar[r]^{\qquad\sim}_{\qquad T_1} & K^1(\nil_k) \ar[d]^{\iota^*} \\
K^0(S^1 \times \TT) \ar[r]^{\sim}_{T_c} & K^1(S^1 \times \TT)  }  \nonumber
\eeq
and the factorisations $T_2=T\circ T_1^{-1}$ and $T_{ac}=T_a\circ T_c^{-1}$; in other words, the combined diagram \eqref{BBfull} commutes.

\begin{remark}
For the degree-shifted versions of \eqref{BB} and \eqref{BB2}, a torsion subgroup $\ZZ_k$ appears in the two $K$-theory groups on the top rows (when $k\neq 1$). We can ignore these for the commutativity of \eqref{BB} since the $K$-theory groups in the bottom rows are torsion-free; the use of cyclic theory continues to work in this case, but we have left out the explicit computations.
\end{remark}

\section{Topological insulators with screw dislocations and the Heisenberg group}\label{section:screw}
In \cite{RZV}, it was proposed that screw dislocations (Fig.\ \ref{fig:screw}) in a 3D time-reversal invariant topological insulator can can host topologically protected modes traversing the bulk bandgap. This phenomenon was demonstrated numerically in a tight-binding model, where a pair of screw dislocations was introduced into a unit cell subjected to periodic boundary conditions. A subsequent investigation into such dislocation-bound 1D modes can be found in \cite{ITT}.

Prior to these papers, the effect of a single screw dislocation on Landau levels was studied in \cite{Furtado}, using a differential-geometric theory of defects as described in \cite{BBS,KV,Kleinert}. The latter framework is very general, and allowed the authors of \cite{SF} to consider the quantum dynamics of a free particle in the presence of a cylindrically symmetric \emph{distribution} of parallel screw dislocations. It was found that the energy levels were quantized to \emph{elastic Landau levels} reminiscent of those of a charged particle in the presence of a uniform magnetic field.

In general, defects such as screw dislocations break the $\ZZ^3$ translation symmetry of the original 3D Euclidean lattice, which is needed to define the Fu--Kane--Mele $\ZZ_2$-invariants \cite{FKM} characterizing a 3D time-reversal invariant topological insulator. This raises the question as to whether the standard ``undeformed'' Fu--Kane--Mele invariants continue to be appropriate. Furthermore, the notion of a unit cell in $\RR^3$, used in \cite{RZV}, also requires this symmetry to be well-defined. We can circumvent this general difficulty by requiring the defects themselves to be distributed in a sufficiently regular manner (this is implicit in the imposition of periodic boundary conditions in \cite{RZV,ITT}). This regularity allows us to define topological invariants which are, in a precise sense, deformed versions of the usual ones arising from (commutative) Bloch theory. We explain this deformation in the simpler case without time-reversal symmetry, leaving the time-reversal invariant case for a subsequent work.

\begin{figure}[h]
    \centering
    \includegraphics[width=0.95\textwidth]{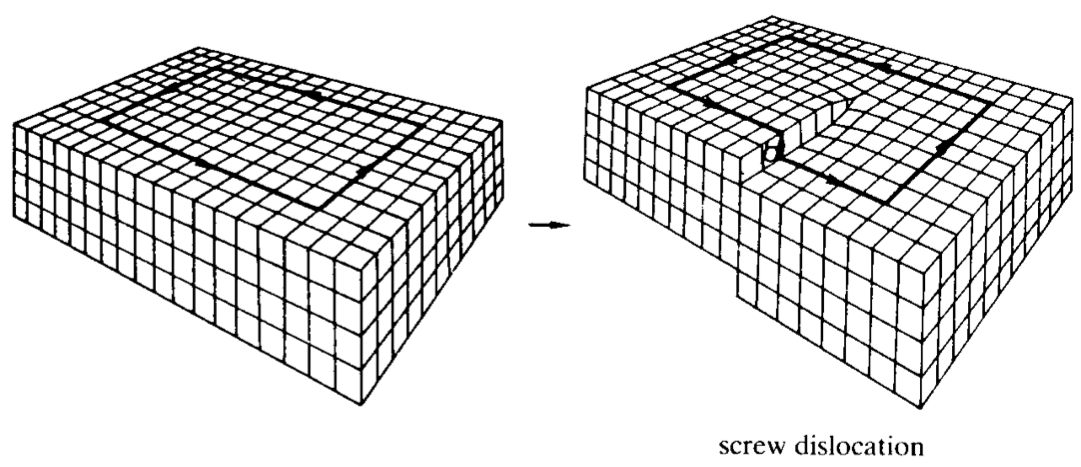}
    \caption{An elementary screw dislocation (Source: pp.\ 786 in \cite{Kleinert}) with Burgers vector in the vertical direction. A circuit of translations in the horizontal direction enclosing the dislocation ends at a lattice site which differs from the starting point by a vertical translation.}
    \label{fig:screw}
\end{figure}

Consider at first a 3D lattice of atomic sites in $\RR^3$ with lattice translations along axes labelled by $\aR, \bR, \cR$. Then introduce parallel elementary screw dislocations along the $\cR$-direction, which are distributed uniformly so that the dislocations are located on a 2D lattice when projected onto the $\aR$-$\bR$ plane. Let $U$ and $V$ be generating translations for this latter lattice in the $a$ and $b$ directions respectively, taken along the atomic bonds in the \emph{distorted} 3D lattice. Then $U$ and $V$ no longer commute, but instead obey $UV=WVU$ with $W$ being a translation in the $c$-direction by one atom --- these are the defining commutation relations for the integer Heisenberg group. Thus one may consider for instance, a tight-binding Hamiltonian on the above distorted 3D lattice symmetric under ``translations'' by $\heisZ$ as above.

\subsection{$\heisZ$-symmetric Hamiltonians and the noncommutative Brillouin zone}
Ordinary lattices in $\RR^3$ have an abelian group $\ZZ^3$ of translation symmetries which allows the use of Bloch theory for defining topological invariants of gapped Hamiltonians with these symmetries. In this commutative case, the Brillouin zone is a 3-torus which is the Pontryagin dual of $\ZZ^3$, and the topological invariants live in various $K$-theory groups of this torus \cite{FM}.

A noncommutative generalization is needed for the integer quantum Hall effect, where the \emph{magnetic} translation symmetries do not commute but instead generate a noncommutative torus \cite{Bellissard}, which is a deformation of the ordinary 2-torus. The analysis of Hamiltonians symmetric under $\heisZ$ follows along much the same lines. Namely, the ``noncommutative Brillouin zone'' is the integer Heisenberg group $C^*$-algebra $C^*(\heisZ)$, which as we have already seen in Section \ref{section:Heisenberggeneralities}, is a parametrised deformation of the ordinary 3-torus.

In the ordinary commutative case, the 3-torus has three independent first Chern classes, represented by the Bott projections $P_{ac}, P_{bc}, P_{ab}$ which are each non-trivial on one of the three independent choices of 2-subtori. Each of these Chern classes corresponds to a 3D Chern insulator which can be thought of as layers of a standard 2D Chern insulator. The quantum Hall state in 3D is again characterised by three integers \cite{KHW}, which come from the noncommutative Chern classes for the 3D noncommutative torus.

However, when we consider $\heisZ$ symmetry, the spatial direction along $\cR$ is singled out, and we see its signature in the ``disappearance'' of the transverse Bott projection $P_{ab}$ from the generators of the $K$-theory of $C^*(\heisZ)$ calculated in Section \ref{section:Heisenbergalgebraktheory}. Furthermore, the bulk-boundary homomorphism $\partial$ maps only onto $\ZZ[W]\in K_1(C(S^1\times\aT))=K_1(C^*(\acZ))$ but not in $\ZZ[U]$ (if we took the boundary to be transverse to the $a$-direction, then $\partial$ will again land in $\ZZ[W]$ but not in $\ZZ[V]$). Recall that $W$ is the generator of translations in the $c$-direction. This suggests that edge modes propagate only in the $c$-direction, which is consistent with the intuition in \cite{RZV,ITT} that topologically protected 1D propagating modes develop along screw dislocations.

\begin{remark}
The continuum limit of a uniform distribution of parallel screw dislocations was studied in \cite{CK}, leading to the real Heisenberg group manifold $\heisR$ (which is topologically still $\RR^3$) as the so-called \emph{material manifold}. The discrete subgroup $\heisZ$ is a lattice in $\heisR$, and the quotient nilmanifold $\heisR/\heisZ$ is the appropriate fundamental domain, or ``Wigner--Seitz cell'', to use here.

The restriction map in \eqref{BB2}, which is the T-dualized bulk-boundary homomorphism, now has a direct interpretation: it is a simple restriction of $K$-theory classes from the deformed bulk fundamental domain (a nilmanifold) to the boundary fundamental domain (a torus).
\end{remark}

\appendix
\section{Cyclic cocycles on integer Heisenberg group algebra and pairing with $K$-theory}
\subsection{Line bundles over $\nil_k$}\label{section:linebundlesovernil}
The group cohomology $H^2_{\mathrm{group}}(\heisZ(k);\ZZ)\cong H^2(\nil_k;\ZZ)=\ZZ^2\oplus\ZZ_k$ is generated by the following central extensions of $\heisZ(k)$ by $\ZZ$.
First, there is
\begin{equation}
\widetilde{\heisZ(k)}^{bc}\coloneqq\left\{\begin{pmatrix}1 & a & \frac{c}{k} & \frac{d}{2k} \\ 0 & 1 & b & \frac{1}{2}b^2 \\ 0 & 0 & 1 & b \\ 0 & 0 & 0 & 1\end{pmatrix}\Big|\, a,b,c,d\in\ZZ\right\}\nonumber
\end{equation}
with quotient map onto $\heisZ$ being restriction to the upper left $3\times 3$ submatrix. Second, there is
\begin{equation}
\widetilde{\heisZ(k)}^{ac}\coloneqq\left\{\begin{pmatrix}1 & a & \frac{1}{2}a^2 & \frac{d}{2k} \\ 0 & 1 & a & \frac{c}{k} \\ 0 & 0 & 1 & b \\ 0 & 0 & 0 & 1\end{pmatrix}\Big|\, a,b,c,d\in\ZZ\right\}\nonumber
\end{equation}
with quotient map being restriction to the lower right $3\times 3$ submatrix.
Third, there is
\begin{equation}
\widetilde{\heisZ(k)}^{ab}\coloneqq\left\{\begin{pmatrix}1 & a & \frac{c}{k} & d \\ 0 & 1 & b & b \\ 0 & 0 & 1 & 0 \\ 0 & 0 & 0 & 1\end{pmatrix}\Big|\, a,b,c,d\in\ZZ\right\}.\nonumber
\end{equation}
with quotient map being restriction to the upper left $3\times 3$ submatrix.

Their group cocycles are, respectively,
\begin{align}
-\sigma_{bc}((a_1,b_1,c_1),(a_2,b_2,c_2)) &= 2c_1b_2+ka_1b_2^2\sim -b_1c_2+c_1b_2-ka_1b_1b_2\nonumber\\
\sigma_{ac}((a_1,b_1,c_1),(a_2,b_2,c_2)) &=2a_1c_2+ka_1^2b_2\sim a_1c_2-c_1a_2-ka_1a_2b_2\nonumber\\
\sigma_{ab}((a_1,b_1,c_1),(a_2,b_2,c_2)) &=a_1b_2.\nonumber
\end{align}
where $\sim$ means ``cohomologous to''. The cocycles $\sigma_{bc}$ and $\sigma_{ac}$ are the free generators, whereas one can check that $k\cdot\sigma_{ab}$ is a coboundary so that $\sigma_{ab}$ generates $\ZZ_k$. Furthermore, we observe that $\sigma_{bc}$ vanishes on $\acZ$ while $\sigma_{ac}$ vanishes on $\bcZ$.

These constructions also give rise to central extensions of $\heisR$ by $\RR$ by allowing $a,b,c,d\in\RR$. The non-trivial line bundles over $\nil_k=B\heisZ(k)$ are obtained by quotienting these real extensions by the discrete ones.

\subsection{Cyclic cocycles from group cocycles}\label{section:cycliccocycles}
According to Connes--Moscovici \cite{CM}, we can construct cyclic 2-cocycles $\tau_{bc}, \tau_{ac}$ on (a smooth subalgebra of) $C^*(\heisZ(k))$ from the group 2-cocycles $\sigma_{bc}, \sigma_{ac}$ on $\heisZ(k)$. Explicitly, $\tau_{bc}$ is
\begin{equation}
\tau_{bc}(f_0,f_1,f_2)=\sum_{\substack{\gamma_0\gamma_1\gamma_2=\mathrm{id}\\\gamma_0,\gamma_1,\gamma_2\in\heisZ(k)}}\sigma_{bc}(\gamma_1,\gamma_2)f_0(\gamma_0)f_1(\gamma_1)f_2(\gamma_2),\qquad f_i\in \CC(\heisZ(k)),\label{cyclicfromgroupcocycle}
\end{equation}
extended to the smooth subalgebra of $C^*(\heisZ(k))$, and similarly for $\tau_{ac}$. These cyclic cocycles then pair with $K_0(C^*(\heisZ(k)))$ in the usual way, by extending the formula \eqref{cyclicfromgroupcocycle} to smooth matrix algebras over $\CC(\heisZ(k))$,
\begin{equation}
\tau_{bc}(f_0\otimes A_0,f_1\otimes A_1,f_2\otimes A_2)=\mathrm{tr}(A_0A_1A_2)\sum_{\substack{\gamma_0\gamma_1\gamma_2=\mathrm{id}\\\gamma_0,\gamma_1,\gamma_2\in\heisZ(k)}}\sigma_{bc}(\gamma_1,\gamma_2)f_0(\gamma_0)f_1(\gamma_1)f_2(\gamma_2),\label{cyclicfromgroupcocyclematrix}
\end{equation}
and taking $$\langle\tau_{bc},[P] \rangle=\tau_{bc}(P,P,P),\qquad [P]\in K_0(C^*(\heisZ(k))).$$ The aim is to show that $\tau_{bc}, \tau_{ac}$ are linearly independent and non-zero. Note that $\sigma_{ab}$ is torsion, so we do not get another cyclic cocycle $\tau_{ab}$ from it.

Corresponding to the two commutative subalgebras $C^*(\bcZ), C^*(\acZ)$, there are two projections $P_{bc}, P_{ac}$ (in their matrix algebras) which are the Bott projections $P_\mathrm{Bott}$ when $C^*(\bcZ)$ and $C^*(\acZ)$ are identified with $C(\TT^2)$. As an element of $C^*(\heisZ(k))$, $P_{bc}$ can be thought of as the function $\aZ\rightarrow C^*(\bcZ)$ which is $P_\mathrm{Bott}$ when $a=0\in\aZ$ and zero otherwise; similarly for $P_{ac}$.

Recall that there is also a standard cyclic 2-cocycle $\psi$ on $C^\infty(\TT^2)$, defined by
\begin{equation}
\psi(f,g,h)=2\pi i\,\tau(f[\partial_1g,\partial_2h]),\qquad f,g,h\in C^\infty(\TT^2),\nonumber
\end{equation}
where $\tau(\sum_{m,n\in\ZZ}a_{m,n}U^mW^n)=a_{0,0}$ with $U,W$ the two commuting unitaries generating $C^\infty(\TT^2)$. The derivations are such that $\partial_1(U)=U, \partial_2{W}=W$, and $\partial_1(W)=0=\partial_2(U)$. It is known that the pairing of $\psi$ with $K_0(C(\TT^2))$ is given by 
\begin{equation}
\langle \psi,[P_\mathrm{Bott}] \rangle=1,\qquad\langle \psi,[{\bf 1}]\rangle=0,\nonumber
\end{equation}
and that $\tau(P_\mathrm{Bott})=1=\tau({\bf 1})$. 

We can rewrite the formula for $\psi$ in a form which resembles \eqref{cyclicfromgroupcocycle}. Let 
\begin{align}
f &=\sum_{m,n\in\ZZ}f_{m,n}U^mW^n,\nonumber\\
g &=\sum_{p,q\in\ZZ}g_{p,q}U^pW^q,\nonumber\\
h &=\sum_{r,s\in\ZZ}h_{r,s}U^rW^s,\nonumber
\end{align}
then
\begin{equation}
\psi(f,g,h)=\sum_{p,q,r,s\in\ZZ}(ps-rq)\cdot f_{-p-r,-q-s}g_{p,q}h_{r,s},\label{intermediatepsi}
\end{equation}
which is then extended as in \eqref{cyclicfromgroupcocyclematrix} to a pairing with $K_0(C(\TT^2))$. Let $\gamma_0\equiv(m,n), \gamma_1\equiv(p,q),\gamma_2\equiv(r,s)$ be elements of $\ZZ^2$, so $c(\gamma_1,\gamma_2)=ps-rq$ defines a 2-cocycle on $\ZZ^2\times\ZZ^2$. Equation \eqref{intermediatepsi} can be rewritten as
\begin{equation}
\psi(f,g,h)=\sum_{\substack{\gamma_0\gamma_1\gamma_2=\mathrm{id}\\\gamma_0,\gamma_1,\gamma_2\in\ZZ^2}}c(\gamma_1,\gamma_2)f(\gamma_0)g(\gamma_1)h(\gamma_2),\nonumber
\end{equation}
and the statement that $\langle \psi,[P_\mathrm{Bott}]\rangle=1$ becomes
\begin{equation}
\langle \psi,[P_\mathrm{Bott}]\rangle =\sum_{\substack{\gamma_0\gamma_1\gamma_2=\mathrm{id}\\\gamma_0,\gamma_1,\gamma_2\in\ZZ^2}}c(\gamma_1,\gamma_2)\,\mathrm{tr}\left(P_\mathrm{Bott}(\gamma_0)P_\mathrm{Bott}(\gamma_1)P_\mathrm{Bott}(\gamma_2)\right)=1.\label{pairingwithBott}
\end{equation}

\subsection{Pairing with $K$-theory}
We can now compute the pairing of $\tau_{bc}$ with $[P_{ac}]$ and $[P_{bc}]$:
\begin{align}
\tau_{bc}(P_{ac},P_{ac},P_{ac}) &=\sum_{\substack{\gamma_0\gamma_1\gamma_2=\mathrm{id}\\ \gamma_0,\gamma_1,\gamma_2\in\heisZ(k)}}\sigma_{bc}(\gamma_1,\gamma_2)\,\mathrm{tr}\left(P_{ac}(\gamma_0)P_{ac}(\gamma_1)P_{ac}(\gamma_2)\right)\nonumber\\
&= \sum_{\substack{\gamma_0\gamma_1\gamma_2=\mathrm{id}\\ \gamma_0,\gamma_1,\gamma_2\in\acZ}}\sigma_{bc}(\gamma_1,\gamma_2)\,\mathrm{tr}\left(P_{ac}(\gamma_0)P_{ac}(\gamma_1)P_{ac}(\gamma_2)\right),\nonumber
\end{align}
since $P_{ac}$ is supported on $\acZ$. But $\sigma_{bc}$ restricted to $(\gamma_1,\gamma_2)\in\acZ\times\acZ$ vanishes, as we had found at the end of the previous subsection, so we conclude that
\begin{equation}
\langle\tau_{bc},[P_{ac}]\rangle\equiv\tau_{bc}(P_{ac},P_{ac},P_{ac})=0.\nonumber
\end{equation}
Next, 
\begin{align}
\langle\tau_{bc},[P_{bc}]\rangle\equiv\tau_{bc}(P_{bc},P_{bc},P_{bc}) &=\sum_{\substack{\gamma_0\gamma_1\gamma_2=\mathrm{id}\\ \gamma_0,\gamma_1,\gamma_2\in\heisZ(k)}}\sigma_{bc}(\gamma_1,\gamma_2)\,\mathrm{tr}\left(P_{bc}(\gamma_0)P_{bc}(\gamma_1)P_{bc}(\gamma_2)\right)\nonumber\\
&= \sum_{\substack{\gamma_0\gamma_1\gamma_2=\mathrm{id}\\ \gamma_0,\gamma_1,\gamma_2\in\bcZ}}\sigma_{bc}(\gamma_1,\gamma_2)\,\mathrm{tr}\left(P_{bc}(\gamma_0)P_{bc}(\gamma_1)P_{bc}(\gamma_2)\right)\nonumber\\
&= \sum_{\substack{\gamma_0\gamma_1\gamma_2=\mathrm{id}\\ \gamma_0,\gamma_1,\gamma_2\in\ZZ^2}}c(\gamma_1,\gamma_2)\,\mathrm{tr}\left(P_\mathrm{Bott}(\gamma_0)P_\mathrm{Bott}(\gamma_1)P_\mathrm{Bott}(\gamma_2)\right)\nonumber\\
&= 1,\nonumber
\end{align}
where the second equality follows from the fact that $P_{bc}$ is supported on $\bcZ$ (where it is $P_\mathrm{Bott}$), the third equality follows from $\sigma_{bc}|_{\bcZ\times\bcZ}=c$, and the fourth equality is \eqref{pairingwithBott}.

To summarize,

\begin{proposition}\label{GCT}
\begin{equation}
\langle\tau_{bc},[P_{bc}]\rangle=1,\qquad \langle\tau_{bc},[P_{ac}]\rangle=0,\qquad
\langle\tau_{bc},[{\bf 1}]\rangle=0.\nonumber
\end{equation}
In a similar vein, we also obtain
\begin{equation}
\langle\tau_{ac},[P_{ac}]\rangle=1,\qquad \langle\tau_{ac},[P_{bc}]\rangle=0,\qquad
\langle\tau_{ac},[{\bf 1}]\rangle=0.\nonumber
\end{equation}
It is known, from \cite{AP,Ko} for instance, that $K_0(C^*(\heisZ(k)))\cong\ZZ^3$ is generated by $[P_{ac}], [P_{bc}]$, and the trivial projection $[{\bf 1}]$ (the identity). We can define the 0-cocycle $\tau$ on the smooth subalgebra of $C^*(\heisZ(k))$ by $\tau(\sum_{r,s,t\in\ZZ}a_{r,s,t}U^rV^sW^t)=a_{0,0,0}$, whence we see that
\begin{equation}
\langle\tau,[P_{ac}]\rangle=\langle\tau,[P_{bc}]\rangle=\langle\tau,[{\bf 1}]\rangle=1.\nonumber
\end{equation}
Thus $\tau_{bc}, \tau_{ac}$ are linearly independent, and together with $\tau$, can be used to distinguish elements of $K_0(C^*(\heisZ(k)))$ uniquely.
\end{proposition}

\begin{remark}
Hadfield \cite{Hadfield} also studies these pairings in the $k=1$ case, however his calculations there are not complete, and do not exploit the simplifications that we do.

\end{remark}

\bigskip
\noindent{\it Acknowledgements}. This work was supported by the Australian Research Council via ARC Discovery Project grants 
DP150100008 and DP130103924.



\end{document}